\documentclass[11pt]{article}
\usepackage{amsmath}
\usepackage{amsthm}
\usepackage{amsfonts}
\usepackage{amssymb}
\usepackage[margin=1.5in]{geometry}
\theoremstyle{definition}
\usepackage{graphicx}

\usepackage{caption}
\usepackage{subcaption}
\usepackage{hyperref}
\usepackage{mathtools}
\mathtoolsset{showonlyrefs}
\newtheorem{theorem}{Theorem}[section]

\newtheorem{definition}[theorem]{Definition}
\newtheorem{example}[theorem]{Example}
\newtheorem{lemma}[theorem]{Lemma}
\newtheorem{proposition}[theorem]{Proposition}
\newtheorem{remark}[theorem]{Remark}

\usepackage{enumitem}
\numberwithin{equation}{section}
\usepackage{color}
\usepackage{morefloats}
\usepackage{natbib}

\title{Criticality and Popularity in Social Networks}

\author{Eberhard Mayerhofer\thanks{University of Limerick, Department of Mathematics and Statistics, Castletroy, Co. Limerick, Ireland, email {eberhard.mayerhofer@ul.ie}.}}

\begin{document}
\maketitle
\thispagestyle{empty}

\vspace{-1cm}
\begin{abstract}
I find that several models for information sharing in social networks can be interpreted as age-dependent multi-type branching processes, and build them independently following Sewastjanow. This allows to characterize criticality in (real and random) social networks. For random networks, I develop a moment-closure method that handles the high-dimensionality of these models: By modifying the timing of sharing with followers, all users can be represented by a single representative, while leaving the total progeny unchanged. Thus I compute the exact popularity distribution, revealing a viral character of critical models expressed by fat tails of order minus three half.
\end{abstract}

\smallskip


\maketitle
\newpage
\tableofcontents

\section*{Introduction}
Modern communication is facilitated by social media on the world wide web, where, thanks to computer technology, information may be shared instantly between a sizeable portion of the earth's population.\footnote{\url{ourworldindata.org/rise-of-social-media}} At times, users' attention seemingly explodes, by sharing quickly among connected members, thus reaching numerous readers, a situation that is typically refereed as ``viral". Empirical research on the dynamics of information spread on the world wide web outweighs fundamental research using mathematical tools. This paper develops probabilistic models of information flow through a social media network. We take inspiration from several models \cite{gleeson2014competition, gleeson2016effects, d2019spreading}, that find that a branching mechanism may explain several empirically observed traits of information spreading\footnote{These are analytically tractable models, using ideas of the empirical paper \cite{weng2012competition}.)} Users of one or several interconnected social media apps share information with their followers, either by creating new threads or sharing existing ones. A large amount of information received by users competes for their attention.  The branching assumption implies that users treat each incoming piece of information (henceforth called meme\footnote{According to Miriam-Webster, a meme is an idea, behaviour, style, or usage that spreads from person to person within a culture. We use this notion in the same way as \cite{gleeson2014competition} to identify information of the same or similar content to be able to book-keep the dynamics of a piece of information.}) equally, and their future evolution is independent of each other, so that the popularity of a meme is summarized by the total number of sharing of (or responses to) a single meme. The `viral' character - in probabilistic terms rare events of large popularity - observed in reality (e.g., hashtags in Twitter \cite{banos} or videos on YouTube \cite{szabo2010predicting}\footnote{For references to further empirical literature, see the citations in \cite{gleeson2016effects} and \cite{d2019spreading}.}) may be replicated in critical or near-critical models by the fat tailed popularity distribution. It goes without saying that a branching mechanism can only serve as an approximation of information spreading on a real social network. For example, all models we discuss below have the common feature that the total information arriving at a node is partly exogenous, as we only keep track of the dynamics of a specific meme shared with followers (the endogenous part), which users select from their individual stream of competing memes (the endogenous part).\footnote{The presentation of a branching processes as a tree may 
lead to the misleading view that the existence of reciprocal links in a real social network is in conflict with a branching model. However, one can model with the same means the communication between two individuals: Obviously, the tree structure representing all communication on a single meme between the two does not imply a unilateral communication.}

The rest of the paper is structured as follows:
\begin{enumerate}
\item  \label{work1} Section \ref{sec: branching} develops the foundations of multi-type age-dependent processes, along the lines of \cite{Sewastjanov}, thereby re-defining the concept of final classes, so to be able to establish the precise conditions under which extinction occurs. (See Remarks \ref{rem: final classes} and \ref{remk: age}, and Theorem \ref{th: irr char disc timec} with proof.) This is contrary to \cite{gleeson2014competition, gleeson2016effects, d2019spreading} who exclusively use univariate PGFs and some approximations to handle multi-type processes. 

\item \label{work2} Next, we put forward a new setup of these models which is consistent with the branching processes literature in Section \ref{sec: Models}, using merely the prosaic explanation found in the original papers. This undertaking
is  motivated by \cite{d2019spreading} who note that ``this model can be thought of as an age-dependent multi-type branching process" in the sense of the monograph \cite{athreya1972peter}. However,  we found the more general class of age-dependent branching processes of \cite{Sewastjanov} more suitable to replicate the competition-induced criticality in the sense of \cite{gleeson2014competition}. Also, our approach allows to perfectly match the first moments of \cite{gleeson2014competition, gleeson2016effects, d2019spreading}, and thus to appeal to standard results  what concerns extinction (see also Section \ref{sec: branching}) or popularity (see the end of Section \ref{sec: epsilon 0 representation}).

We consider two main model classes:  A model of information spreading on an actual (possibly multi-layer) network, where each layer represents a social media app (Model 1 and 1b), and the other one (Model 2) builds on a random network, so to reduce dimensionality of the former. Unlike the intuitive derivation of the key probability generating functions in \cite{gleeson2014competition, gleeson2016effects, d2019spreading}, we obtain the delay integral equations governing multi-type branching models, using only the classical basic building blocks for the branching mechanisms (timing of particles' death, and the distributions of descendants upon death of a particle), and then appeal to the results developed in Section \ref{sec: branching} to characterize critical behaviour. We then prove the conjecture of \cite{d2019spreading} concerning (sub)criticality of the Model 2. To this end we analyze in Section \ref{sec: Matrices} the spectrum of so-called Skew-sub-stochastic matrices, which, in the irreducible case, surprisingly constitute the class of non-negative matrices with spectral radius $\leq 1$ (See, e.g. Theorem \ref{thm: skew stoch}). \footnote{\cite{d2019spreading} proves criticality under three additional assumptions: smallness of certain parameters (innovation),  irreducibility of the first moment matrix, and dominance of one layer.} We are thus able to answer the conjecture of \cite[Section 3.2]{d2019spreading} in the positive, that ``the system is subcritical for all valid parameter values".\footnote{ This conjecture was hardened by\cite{d2019spreading} through simulations.} (See Theorem \ref{th: crit irr etc}, Theorem \ref{prop: subskew social} and Section \ref{sec: multi} for the multi-layer case.) Remark \ref{rem: parameterization} reflects on the maximal parameterization of Model 1.
\item \label{work3} Moment-closure, as introduced in \cite{gleeson2014competition}, but also implicitly used in \cite{gleeson2016effects}, can be interpreted as means to represent a multi-type branching process by a single-type processes,
so to make the usually high-dimensional problems (e.g., characterising criticality and estimating the popularity distribution) analytically tractable. In Section \ref{sec: epsilon 0 representation} we study this technique for random networks.  To answer the remaining problem of \cite{gleeson2014competition} concerning the quality of approximation, we develop an exact closure method that allows to compute the popularity distribution, revealing
fat tails of order -$3/2$ in the critical case and thus agrees, modulo a normalizing constant, with the aforementioned empirical findings, as well as the approximation of \cite{gleeson2014competition}.
\end{enumerate}

\section{Multi-type Age-dependent Branching processes}\label{sec: branching}

In this section we summarize a few fundamental statements about age-dependent branching processes with multiple types $T^1,\dots,T^n$. Each particle of type $T^i$ lives a random time $\tau^i$ with distribution function $G_i(t)=\mathbb P[\tau^i\leq t]$.

Conditioned on the event $\{\tau^i=u\}$, the probability generating function (henceforth PGF) of the particle distribution $\nu^i$ (the totality of particles of each type, emerging when one $T^i$ particle dies) is given by
\[
h^i(s;u):=\mathbb E[s^{\nu^i}\mid \tau^i=u]=\sum_{\alpha\in \mathbb N_0^n }p_\alpha^i(u) s^\alpha,
\]
where $s=(s^1,\dots,s^n)$ is the argument of the PGF, and we use the notation $s^\alpha=(s^1)^{\alpha_1}\dots (s^n)^{\alpha_n}$

In the typical definition of age-dependent branching processes  the conditional probabilities $p^i_\alpha(u)$ do not depend on the age $u$ of the article (cf.~\cite[Chapter 28.3, p.158]{harris1963theory} \cite[p.225]{athreya1972peter},  or \cite{goldstein1971critical}). As we need age-dependence in this sense, we use the setup of the monograph \cite{Sewastjanov}. As the book is only available in the Russian original, or its German translation by Uwe Prehn, we give a summary of the essential references.

For $1\leq i\leq n$, the vector $\mu^i(t)=(\mu^i_1(t),\dots\mu^i_n(t))$ describes the number of particles $\mu^i_j(t)$, assuming the process has started with a single individual of type $T^i$.
We define the PGFs
\begin{equation}\label{eq; Fi}
F^i(t,s)=\mathbb E[s^{\mu^i(t)}], \quad 1\leq i\leq n.
\end{equation}

The branching mechanism implies that the distribution $\mu(t)\mid \mu(0)=\beta$, where $\beta=(\beta_1,\dots,\beta_n)$ is fully specified by the $\mu^i$'s, (and thus the $h^i$'s), as the evolution
of any two particles of any type is mutually independent. In other words,
\[
\mathbb E[s^{\mu(t)}\mid \mu(0)=(\beta_1,\dots,\beta_n)]=(F^1(s))^{\beta_1}\dots (F^n(s))^{\beta_n}.
\]
We use the abbreviation $F(t,s):=(F^1(t,s),\dots, F^n(t,s))$.

Conditioning on the time of death of each initial particle, and using the law of total expectation as well as the branching property, we obtain (\cite[Proof of Satz VIII.1.1]{Sewastjanov}):
\begin{theorem}\label{th: unique existence}
The function $F(t,s)$ satisfies the system of equations 
\begin{equation}\label{eq: superdyn}
F^i(t,s)=\int_0^t h^i(u, F(t-u,s)dG^i(u)+s^i (1-G^i(t)),\quad t\geq 0,\quad 1\leq i\leq n.
\end{equation}
\end{theorem}

\subsection{Final Classes}

In this section we introduce the notion classes, and re-define the notion of final classes for age-dependent processes. By averaging over age, we obtain the unconditional particle distribution upon death of a single particle of type $T^i$, defined by $p^i_\alpha:=\mathbb P[\nu^i=\alpha]$, $1\leq i\leq n$ in terms of the following PGFs
\begin{equation}\label{eq: age out}
h^i(s):=\int_0^\infty h^i(u,s)dG^i(u),\quad s\in [0,1]^n.
\end{equation}
We start with the following:
\begin{definition}\label{def: final class}
We say that Type $T_k$ follows type $T_i$, or that $T_i$ precedes $T_k$ and write $T_i\rightarrow T_k$, if there exists $t>0$ such that
\[
\mathbb P[\mu^i_k(t)>0]>0.
\]
\end{definition}

A class $S_{i0}\subseteq \{ T^1,\dots, T^m\}$, where $1\leq i\leq n$, is the collection of particles that both precede and follow particle type $T_i$.

\begin{example}\label{ex sew}
We now expand a little bit on the example found in \cite[Chapter IV.6]{Sewastjanov}: Let $n=7$, and the branching process be defined via the generating functions
\begin{align*}
h^1(s)&=\frac{s^1(s^2+s^3s^5)}{2},\quad h^2(s)=\frac{(s^2)^2+1}{2},\\
h^3(s)&=s^4,\quad h^4(s)=s^3,\quad h^5(s)=s^6,\\
h^6(s)&=\frac{3s^6+s^7}{4},\quad h^7(s)=\frac{s^6+2s^7}{3}.
\end{align*}
From $h^1$ we see that particles of type $T^1,T^2,T^3$ and $T^5$ follow $T^1$ in a single step (and therefore all particles, check the other generating functions). But from the other generating functions we see that no particles precede $T^1$, except $T^1$. Therefore 
\[
K_1:=S_{10}=\{T_1\}.
\]
From $h^2$ we see only type $T^2$ follows $T^2$, and therefore only $T^2$ precedes $T^2$, and we get
\[
K_2:=S_{20}=\{T_2\}.
\]
From $h^3$ we see that Type $T^3$ is followed by $T^4$, and from $h^4$ we see that $T^4$ is followed by $T^3$, and therefore
\[
K_3:=S_{30}=S_{40}=\{T_3,T_4\}.
\]
Similarly, we see that 
\[
K_4=S_{60}=S_{70}=\{T_6,T_7\}.
\]
The remaining class is  singular, as $S_{50}=\{\}$: From $h^5$ we see that Type $T^6$ follows $T^5$, and by $h^6$ we see that $T^7$ follows $T^6$, but the former Type can only  produce particles of type $T^6$ (see $h^7$.). On the other hand, only type $T^1$ precedes $T^5$, see $h^1$.
In this case we define
\[
K_5=\{T_5\}.
\]
(that the subscript $5$ being the same is a pure coincidence.

We see that the classes are pairwise disjoint and their union yields the total set of particles.
\end{example}

\begin{definition}\label{def: final}
A class $K=\{T^{i_1},\dots, T^{i_m}\}$ is called final class, if it is is not singular, and if there exists $t>0$ such that for any $T^{i_j}\in K$ the generating function $h^{i_j}$ is a linear form in the variables $s^{i_1},\dots, s^{i_m}$, that is
\[
h^{i_j}(s)=\sum_{k=1}^m \varphi_{jk}(s) s^{i_k},\quad 1\leq j\leq m,
\]
where the functions $\varphi_{jk}$ do not depend on $s^{i_1},\dots, s^{i_m}$.
\end{definition}

\begin{remark}\label{rem: final classes}
\cite[Definition IV.6.8]{Sewastjanov} defines final classes only for discrete-time and continuous branching processes, and he does so on the stochastic process level, that is, in terms of $F=(F^1,\dots ,F^n$ defined in \eqref{eq; Fi}\footnote{ Meaning, a class $K=\{T^{i_1},\dots, T^{i_m}\}$ is called final class, if it is is not singular, and if there exists $t>0$ such that for any $T^{i_j}\in K$ the generating function $F^{i_j}$ is a linear form in the variables $s^{i_1},\dots, s^{i_m}$, that is
\[
F^{i_j}(t;s)=\sum_{k=1}^m \varphi_{jk}(t;s) s^{i_k},
\]
where the functions $\varphi_{jk}$ do not depend on $s^{i_1},\dots, s^{i_m}$.}
and then shows that if the property holds for some $t>0$, so it does for all times $t>0$ (\cite[Satz IV.6.1]{Sewastjanov}). In particular, in the discrete-time case this implies the property for the
$h^i$'s, and thus Sewastjanow's definition is equivalent to ours in the discrete-time case. However, in the continuous-time, Markovian case the proof requires continuity of $F(t,s)$ in time $t>0$, as well as the functional equation $F(t+\tau;s)=F(t,F(\tau,s))$ for $t,\tau>0$. None of these properties
we have available for general age-dependent processes: The functional equation is replaced by the more general system of delay differential equations \eqref{eq: superdyn},
and regularly results are typically available for single type processes only). It therefore comes as surprise that \cite{Sewastjanov} doesn't properly define the notion of final classes for age-dependent processes, but use it for \cite[Satz VIII.3.2]{Sewastjanov} and the subsequent paragraph). We have therefore modified the setup slightly, and leave as conjecture that Definition \eqref{def: final} is equivalent to \cite[Definition IV.6.8]{Sewastjanov} in the general setting of age-dependent processes.
\end{remark}

\begin{example}
To continue Example \ref{ex sew}, all classes are final, except $K_5$ (which is singular) and $K_2$ (since the generating function $h^2$ is quadratic in $s^2$, not linear.).
\end{example}

The main theme in this paper is extinction, which for multi-type processes is defined as follows:
\begin{definition}\label{defex}
The extinction probabilities $q=(q^1,\dots,q^n)$ are defined as
\[
q^i:=\mathbb P[\exists t>0:\; \mu^i(t)=0],\quad 1\leq i\leq n.
\]
We say, the probability of extinction is one, if $q=1$, that is, $q^i=1$ for $1\leq i\leq n$.
\end{definition}
Due to the branching property, the probability of extinction is one if and only if for any initial population
we have $\mathbb P[\exists t>0:\quad \mu(t)=0]=1$. This validates the notion of extinction in Definition \ref{defex}.

\subsection{The Discrete Case}
The discrete-time branching process is obtained, when assuming that 
\begin{itemize}
\item each particle $T^i$ dies at $t=1$ with probability one, that is $G^i(t)=1_{t\geq 1}$, or in other words, $dG^i(u)=\delta_{t=1}(du)$. 
\item $h^i(u,s)=h^i(s)$, $u\geq 0$.
\end{itemize}

Then \eqref{eq: superdyn} yields constant solutions on $[n,n+1)$ defined inductively by $F^i(n;s)$, where
\[
F^i(0,s)=s^i,\quad\text{and for}\quad n\geq 1,\quad  F^i(n,s)=F(F^i(n-1,s)).
\]

Due to Remark \ref{rem: final classes}, the notion of final classes is the same to Sewastjanow's, at least in the discrete-time case. Therefore
we have the following characterisation (cf.~\cite[Satz V.1.5]{Sewastjanov}):
\begin{theorem}\label{th: main discr}
For a branching process $\mu(t)$ in discrete time with first moment matrix $A$ the following are equivalent:
\begin{enumerate}
\item  The probability of extinction is one.
\item There are no final classes, and $\rho(A)\leq 1$.
\end{enumerate}
\end{theorem}

\subsection{The Continuous Case}
We assume that $G^i(t)$ have continuous densities supported on $[0,\infty)$ (whence $G^i(0)=0$), and that
the first unconditional moments $A^i_j$ of the particle distribution $\nu^i$ given by \eqref{eq: moments 1} are finite
for all $1\leq i,j\leq n$. Thus the assumptions of \cite[Theorem VIII.2.1]{Sewastjanov} are satisfied and imply that the PGFs $F=(F^1,\dots, F^n)$ are the unique solution of
the system of delay differential equations \eqref{eq: superdyn}.

\begin{theorem}\label{th: irr char disc timec}
Let $\mu(t)$ be an age-depending branching process, with first moment matrix $A$. The following are equivalent:
\begin{enumerate}
\item  The probability of extinction is one.
\item The process has no final classes, and $\rho(A)\leq 1$.
\end{enumerate}
\end{theorem}
\begin{proof}
For the entire proof we use the short-hand notation $h=(h^1,\dots,h^n)$, where for each $1\leq i\leq n$, $h^i$ is the PGF of the unconditional particle distribution $(\nu^1,\dots,\nu^n)$ as defined in \eqref{eq: age out}.

By \cite[Satz VIII.3.1]{Sewastjanov}, the extinction probabilities $q=(q^1,\dots,q^n)$ are those solutions of the system $h(s)=s$ which are closest to the origin.

On the other hand, we know that the functions $h(s)$ determines the branching mechanism of a discrete-time branching process $\mu_D(t)$, and that the extinction probabilities $q_D=(q_D^1,\dots, q_D^n)$ of this process are also the smallest non-zero solutions of the same equation $h(s)=s$ on the  unit hypercube $0\leq s^i\leq 1$, $1\leq i\leq n$ (\cite[Satz V.1.4]{Sewastjanov}). Furthermore, final classes are defined, both for the original process $\mu(t)$ and the auxiliary discrete-time process $\mu_D(t)$, by the same function $h(s)$. Therefore, by Theorem \ref{th: main discr}, the extinction probabilities of $\mu_D(t)$ are equals $1$, if and only if $\mu_D(t)$ has no final classes, and $\rho(A)\leq 1$. Since they are the first non-zero solution of the equation $h(s)=s$, any of these statements is equivalent to $\mu(t)$ having unit extinction probability.
\end{proof}

\begin{remark}\label{remk: age}
Final classes were not defined in \cite{Sewastjanov} for age-dependent processes, while the book claims the exact same result as Theorem \ref{th: irr char disc timec} (namely
\cite[Satz VIII.3.2]{Sewastjanov}). The reason for providing a proof in these notes is the new definition of final classes on the level of $h$ in \eqref{eq: age out}, and the fact that \cite{Sewastjanov}
only has an incomplete proof sketch. See also Remark \ref{rem: final classes}.
\end{remark}

\subsection{Irreducibility and Criticality}
\begin{definition}\label{def: irr}
A branching process is called irreducible, if all particle types $\{T_1,\dots, T_m\}$ form a class of connected particle types. All other processes are called reducible.
\end{definition}
An irreducible process has only one class. Therefore, final classes can be characterized easily (we skip the simple proof):
\begin{lemma}\label{char: irr final}
For an irreducible process $\mu(t)$, denote by $K=\{T_1,\dots, T_n\}$ its only class. The following are equivalent:
\begin{enumerate}
\item \label{ita1} $K$ is final.
\item  \label{ita2} There exist constants $\varphi_k\in\mathbb R$ for $k=1,\dots,n$ such that the probability generating functions $h^i(s)$ are of the form
\begin{equation}\label{eq: final irr}
h^i(s)=\sum_{k=1}^n \varphi_k s^k, \quad 1 \leq n.
\end{equation}
\end{enumerate}
\end{lemma}
A (not necessarily irreducible) process that satisfies any of the equivalent statements of Lemma \ref{char: irr final}  has the property that with probability one, any particle (of any type) has exactly one child (of some type). Therefore the total population of such a process stays constant, and thus it never becomes extinct

For the rest of this section, need the unconditional moments of the particle distribution $\nu^i$ are finite. The latter is given, in terms of the PGFs \eqref{eq: age out}, 
\begin{equation}\label{eq: moments 1}
A^i_j:=\frac{\partial}{\partial s^j}\int_0^\infty h^i(u,s)dG^i(u)\mid_{s=1}.
\end{equation}
Irreducible matrices can be characterised as follows:
\begin{proposition}\label{th: irr mat char}
Let $A$ be a non-negative matrix. The following are equivalent:
\begin{itemize}
\item[(a)]  $A$ is irreducible.
\item[(b)]  For each $i,j\in \{1,\dots,n\}$ there exists $1\leq t\leq n$ such that $(A^t) ^i_j>0$. 
\end{itemize}
\end{proposition}
\begin{proof}
By \cite[Theorem 2.1]{Plemmons}, the statements are equivalent, when in (b) the condition $1\leq t\leq n$ is replaced by the weaker condition $t\geq 1$.

So it is left to show that $t$ can be chosen such that $t\leq n$. Note that, since $A$ is non-negative, the irreducibility of $A$ is equivalent to the adjacency matrix $B$ of a graph being irreducible, where each element of $B$ is replaced by $1$, if it is strictly positive. Therefore, without loss of generality $B=A$. Now (a) essentially means that there always exists a path from $i$ to $j$, its length being $t$. Assume $t\geq n+1$. In fact, $(A^t) ^i_j>0$ means, by the very definition of matrix product, that there exists a sequence of indices 
\[
i_0=i,i_1,\dots,i_t=j
\]
such that the sequence of edges $i_l i_{l+1}$, $l=0,1,\dots,t$ define a path. The number of nodes $i_1,i_2, i_t$ is thus strictly larger than $n$, and therefore, there exists $r,s>0$ such that $i_{r}=i_{s}$. That means, one can reduce the length of the path at least by size one.  By repeating this argument, we get a path with length $\leq n$.
\end{proof}

\begin{theorem}\label{th: irr char disc time}
Let $\mu(t)$ be a branching process with finite first moments, that is $A^i_j<\infty$ for all $1\leq i,j\leq n$. Then the following are equivalent:
\begin{enumerate}
\item $\mu(t)$ is irreducible,
\item $A$ is irreducible.
\end{enumerate}
\end{theorem}
\begin{proof}
The property of irreducibility is about the transformation of particle types, not the timing of their death. Therefore $\mu(t)$ is irreducible if and only if
the discrete time-process with generating functions $h^i$, $1\leq i\leq n$ is irreducible, and we shall only consider the latter henceforth. Its matrix of first moments is given by
\eqref{eq: moments 1}.

$\mu(t)$ is irreducible if and only if for each pair $i,j$ there exists a $t>0$ such that $\mathbb P[\mu^i_j(t)>0]>0$. Since $\mu^i_j(t)$ is a non-negative random variable, this is equivalent to the statement that for each pair $i\neq j$ there exists a $t>0$ such that $\mathbb E[\mu^i_j(t)]>0$. By \cite[Chapter IV.4]{Sewastjanov}, $\mu^i_j(t)=(A^t)^i_j$,
where $A^t$ is the $t$-th power of the first moment matrix $A$. Hence, $\mathbb P[\mu^i_j(t)>0]>0$ is equivalent to the existence of a sequence $i_0=i, i_{1},\dots, i_t=j$ such that
\[
A^i_{i_1}A^{i_1}_{i_2}\dots A^{i_{t-1}}_k>0.
\]
Then, by the pigeon hole principle, there must exists a $t$ satisfying the same, however at a perhaps earlier time $t\leq n$. Hence, $\mu(t)$ irreducible is equivalent that for any $i,j$
there exists $t\leq n$ such that $(A^t)^i_j>0$. This is, due to Proposition \ref{th: irr mat char} equivalent to $A$ being irreducible.
\end{proof}
Let $A$ be a non-negative, irreducible matrix. Then there is a unique, strictly positive right eigenvector $u$ and a unique strictly positive left eigenvector $v$ 
associated to $\rho(A)$, normalized such that $v^\top u=1$, and $\sum u^i=1$. Criticality is defined as follows:
\begin{definition}\label{def: crit}
Suppose $\mu(t)$ is an age-dependent, irreducible branching process. Then we call $\mu(t)$
\begin{itemize}
\item subcritical, if $\rho(A)<1$,
\item critical, if $\rho(A)=1$ and $\sum_{i,j,k}v^i B^i_{jk}u_j u_k>0$,
\item supercritical, if $\rho(A)>1$.
\end{itemize}
\end{definition}
\begin{remark}\label{rem: irr}
\cite{Sewastjanov} defines critical behaviour for processes in discrete time as in Definition \ref{def: crit}, while for age-dependent processes he requires that the sub-process constructed from non-final particles satisfies $\rho(A)=1$. This surprising conflict is resolved by realizing that in the age-dependent setup, Sewastjanow does not require irreducibility for the definition of subcritical, critical, and supercritical behaviour. As it is not consistent with the notion found earlier in his book (in the context of discrete or continuous-time branching processes), we refrain from using it.                 
\end{remark}

\section{Meme spreading on Social Media Platforms}\label{sec: Models}
\subsection{Model 1: Static Network}\label{Sec: Model1}
We consider a social media network with $n$ users. Each user $1\leq i\leq n$  receives a stream of so-called memes, separated by time stamps, from the accounts she follows. With probability $\lambda_{ki}$, a meme from account $k$ followed is considered interesting enough to enter the stream, but we condition now on this event that it has been deemed interesting. Once it enters the stream, its existence starts. We are studying the evolution not of all memes, but of one special meme, possibly existing in multiplicity on user $i$'s account. Using the language of branching processes, a meme populating user $i$'s account is identified with a particle of type $T^i$. The branching property is imposed, such that two particles of any type evolve independently.

A particle of type $T^i$ dies, when user $i$ decides that the corresponding meme is shared (in which case it is replaced by this new post), or when reading it, the meme is deemed not worth to be shared (as then the chance that one considers it worth to be shared at a later stage is negligible. Alternatively, think of deleting it.). We assume that user $i$ considers the meme not worth to be shared with the so-called innovation probability $\mu_i \in [0,1)$, in which case the user composes an unrelated meme. In this case the aforementioned particle dies without having descendants. However, if it the meme is shared by user $i$, it enters the stream of all her followers, and each of them will accept it into their stream with probability $\lambda_{ij}\in [0,1]$.

We assume that memes enter the stream of user $i$ according to a Poisson process with rate
\begin{equation}\label{eq: ri}
r_i=\beta_i\mu_i+\sum_{k\neq i}\lambda_{ki}\beta_k,
\end{equation}
where $\beta_k$ are the ``activity rates" of user $k$, $1\leq k\leq n$. We have excluded the $i$-th summand $\lambda_{ii}\beta_i=\beta_i$, as we assume that when a meme is shared (and therefore adds to the stream of memes) its ancestor is deleted, and therefore the activity rate $\beta_i$ should not contribute to the total rate $r_i$.\footnote{The particular formula \eqref{eq: ri} is an exogenous assumption of the model and implies (sub)criticality of the model, as seen below.} Therefore, the occupation time of the meme on user $i$'s account is exponentially distributed, $\tau^i_o\sim \mathcal E(r_i)$.

Let $dG^i(u)$ be the activity distribution of user $i$, that is, the distribution of the random time $\tau^i$, where she becomes active and looks at her stream. (In consistency with \eqref{eq: ri}, we shall assume later that $\tau^i\sim\mathcal E(\beta_i)$.). We assume that $\tau^i$ are independent of $\tau^i_o$. We identify the activity rate with
the life time of particle of type $T^i$. Therefore, in the event $\tau^i<\tau^i_o$, the particle dies, and only with probability $(1-\mu_i)$ she shares and only in this scenario can particle $i$ produce descendants.

Thus, conditional on the event $\{\tau^i=u\}$, where $u\geq 0$, we have 
\begin{equation}\label{eq: kappa}
\mathbb P[\tau^i<\tau^i_o\mid \tau^i=u]=\kappa^i(u)=e^{-r_i u}.
\end{equation}

The age-dependent particle distribution of a type $i$ particle, (that is the composition of descendants, conditional that it dies at an age $u\geq 0$)
is therefore determined by the PGFs
\[
h^i(s;u):=\sum_\alpha p^i_\alpha(u) s^\alpha,
\]
where
\begin{equation}\label{eq: super}
h^i(s; u)=\kappa_0^i(u)+\kappa^i(u)s^i \prod_{j\neq i}\left(1-\lambda_{ij}+\lambda_{ij}s^j\right),
\end{equation}
and the coefficient $\kappa^i$ is given by \eqref{eq: kappa}, whereas
\[
\kappa_0^i(u)=\mu_i+(1-\mu_i)(1-e^{-r_i u}).
\]

\begin{remark}
We have stated here two crucial assumptions concerning independence:
\begin{itemize}
\item First, the branching property makes the evolution of a single particle independent from another one, may it be of the same type or not.
\item Second, conditional on the death of particle $i$, the number of immediate descendants of type $j$ is independent of the number of immediate descendants of type $k$ (in fact, they are mutually independent Bernoulli random variables with parameters $\lambda_{ij}$ resp. $\lambda_{ik}$.)
\end{itemize}
\end{remark}

It remains to specify the life-time distribution $G^i(t)$ of a type $T^i$- particle. The joint distribution of descendants $\nu^i$ and the life-time $\tau^i$ is given by
\[
\mathbb P[\tau_i\leq t, \nu^i=\alpha]=\int p_\alpha^i(u)dG_i(u),
\]
so that the marginal distribution $\tau^i$ is given by
\[
\mathbb P[\tau^i\leq t]=\int_0^t dG^i(u).
\]

We compute now first and second moments. For convenience, we introduce $\lambda_{ii}=1$ for $1\leq i\leq n$, as then
\eqref{eq: super} becomes
\begin{equation}\label{eq: superx}
h^i(s; u)=\kappa_0^i(u)+\kappa^i(u) \prod_{j=1}^n\left(1-\lambda_{ij}+\lambda_{ij}s^j\right).
\end{equation}

\begin{lemma}
The matrix $A:=(A_{ij})_{ij}$ of first moments is given by
\begin{equation}\label{eq: A1}
A_{ij}=\int_0^\infty\frac{\partial h^i(s;u)}{\partial s^j} dG^i(u)= \lambda_{ij}\int_0^\infty \kappa^i(u)dG^i(u).
\end{equation}
The matrices $B^i:=(B^i_{kl})_{kl}$ ($1\leq i\leq n$) of second moments is given by
\begin{equation}\label{eq: B2}
B^i_{kl}=\int_0^\infty\frac{\partial^2 h^i(s;u)}{\partial s^k\partial s^l} dG^i(u)=(1-\delta_{kl}) \lambda_{ik}\lambda_{il} \int_0^\infty \kappa^i(u)dG^i(u).
\end{equation}
\end{lemma}
Note that, if $\tau_a^i\sim\mathcal E(\beta_i)$, then we obtain due to \eqref{eq: kappa}\footnote{Note that the ratio expresses the following fact: For two independent Poisson processes with rate $a,b$ then the probability that the first jumps before the second one is $a/(a+b)$, and here the first jump times of these processes would be independent, and exponentially distributed with rates $\beta_i $ and $r_i$, respectively.}
\begin{equation}\label{eq: super ratio}
\int_0^\infty \kappa^i(u)dG^i(u)=\frac{\beta_i}{\beta_i+r_i}=\frac{\beta_i}{\mu_i\beta_i+\sum_{k=1}^n \beta_{k}\lambda_{ki}},
\end{equation}
where we have used the fact that $\lambda_{ii}=1$.

\subsubsection{Skew-(Sub)Stochastic Matrices}\label{sec: Matrices}

By definition, a stochastic matrix  $B:=(b_{ij})_{1\leq i,j\leq n}$ can be obtained by scaling each element $a_{ij}$ of a non-negative matrix $A$ with its corresponding row sum $\sum_k a_{ik}$. In this section, we introduce an unusual class of so-called Skew-(sub)stochastic matrices $A:=(a_{ij})_{1\leq i,j\leq n}$, which, in their simplest form, originate from non-negative matrices whose row elements $a_{ij}$ are scaled by their corresponding column sum $\sum_{k}a_{ki}$. Such matrices arise naturally as moment matrices in \cite{d2019spreading}, where the influence $a_{ij}$ of user $i$ on another user $j$ must be weighted by the influence $\sum_k a_{ki}$ of those accounts accounts $k$ that $i$ follows. The rationale behind this scaling is that, the more accounts user $i$ follows, the more information arrives at her account, whence the less likely it is that user $i$ shares relevant information with follower $j$. 

Skew-(sub)stochastic matrices are similar to (sub)-stochastic matrices, and therefore exhibit similar spectral properties (See Theorem \ref{thm: skew stoch} below). Before we come to that, let us first start with a formal definition:
\begin{definition}\label{def: skew stoch}
A non-negative matrix $B$ is skew-sub-stochastic, if there are $a_{ij}$, $\gamma_i\in (0,1]$ and $G_i\geq 0$ ($1\leq i,j\leq n$) such that
$G_i+\sum_k a_{ki}>0$ for $1\leq i\leq n$, and
\begin{equation}\label{eq: Skew matrix}
b_{ij}=\frac{\gamma_i a_{ij}}{G_i+\sum_{k}a_{ki}}, \quad 1\leq i,j\leq n.
\end{equation}
\end{definition}
$B$ is called skew-stochastic, if it is skew-sub-stochastic, with $\gamma_i=1$, $G_i=0$ for $1\leq i\leq n$.

By definition, any skew-stochastic matrix is skew-sub-stochastic. However, a skew-(sub)stochastic matrix is, in general, not (sub)stochastic. For example, consider
\[
A=\left(\begin{array}{ll} 1 & 3\\ 1 & 0\end{array}\right), \quad \text{then}\quad B=\left(\begin{array}{ll} 1/2 & 3/2\\ 1/3 & 0\end{array}\right).
\]
The new matrix $B$ is neither row-sub-stochastic (its row sums are not bounded by one), nor column-sub-stochastic. In fact, one element of $B$ is strictly larger than $1$, thus does not qualify for a probability. However, the spectrum of $B$ is given by $\sigma(B)=\{-1/2,1\}$ and therefore the spectral radius equals $1$. Let $D=\text{diag}(G_i+\sum_{k}a_{ki}+1_{\text{if}\,\sum_{k}a_{ki}=0})$. Then, $B=D^{-1} A$ is of similar to$B$, as
\[
C=A D^{-1}=\left(\begin{array}{ll}1/2 & 1\\1/2 &0\end{array}\right)
\]
which is a stochastic matrix, whence due to similarity, $\rho(B)=\rho(C)=1$. The spectral property of this example is not artificial, but a general feature (see Figure \ref{figure1} for an illustration):
\begin{theorem}\label{thm: skew stoch}
Let $B=(b_{ij})_{ij}$ be a non-negative $n\times n$ matrix of form \eqref{eq: Skew matrix} with spectral radius $\rho(B)$. The following hold:
\begin{enumerate}
\item \label{Skew0: pos} $\rho(B)$ is an eigenvalue of $B$.
\item  \label{Skew2: pos} $\rho(B)\leq 1$. 
\item  \label{Skew4: pos} Suppose for each index $1\leq i\leq n$ that $\sum_k a_{ki}\neq 0$. If, in addition, for any $i$, one of the following conditions is satisfied, 
\begin{enumerate}
\item \label{yy1} $a_{ii}\neq 0$ and $\gamma_i<1$,
\item \label{yy2} $G_i>0$,
\end{enumerate}
then $\rho(B)<1$.
\item  \label{Skew6: pos}  Suppose $B$ (or, equivalently, $A$) is irreducible. Then for each index $1\leq i\leq n$ that $\sum_k a_{ki}\neq 0$. If, in addition, for some $i\in\{1,\dots,n\}$, one of the following conditions is satisfied, 
\begin{enumerate}
\item \label{yy1} $a_{ii}\neq 0$ and $\gamma_i<1$,
\item \label{yy2} $G_i>0$,
\end{enumerate}
then $\rho(B)<1$.
\item  \label{Skew3: pos} If $B$ is a skew- stochastic matrix, then $\rho(B)=1$.\footnote{Point \eqref{Skew3: pos} and its proof is similar to a statement found in \cite[Appendix A]{d2019spreading}. The rest of Theorem \ref{thm: skew stoch} is new.}
\end{enumerate}
\end{theorem}

\begin{proof}
Proof of \eqref{Skew0: pos}: It is a well-know fact that non-negative matrices have the property that the spectral radius is an eigenvalue (\cite[Chapter 1, Theorem (3.2)(a)]{Plemmons}. For an elegant proof due to Karlin, using properties of the resolvent, see \cite[Lemma in Section 4]{maccluer2000many}.) For the rest of this proof, we use the abbreviation $\rho:=\rho(B)$. 

We prove next \eqref{Skew2: pos}: We can write $B=DA$, where $D$ is is the diagonal matrix with entries $d_{ii}=\gamma_i/(G_i+\sum_k a_{ki})$. Thus, by construction, $B$ is similar to $AD=D^{-1} B D$, which is a sub-stochastic matrix. Due to similarity, $\rho(B)=\rho(AD)\leq 1$. 

Proof of \eqref{Skew4: pos}: Knowing that $\rho(B)\leq 1$ (see \eqref{Skew2: pos}), let us assume, for  a contradiction, $\rho(B)=1$.  Define $B(\rho):=B-\rho I_n$, where $I_n$ is the $n\times n$ unit matrix.  Let $G=\text{diag}(G_i+\sum_k a_{ki}, 1\leq i\leq n)$. The matrix $G$ is invertible, as by assumption $G_i+\sum_k a_{ki}>0$, $1\leq i\leq n$ . Since $B(\rho)$ is singular, so is $GB(\rho)$. Denote its elements by $d_{ij}$. For each $1\leq i\leq n$, the diagonal element $d_{ii}$ satisfies
\begin{equation}\label{eq: id1}
\vert d_{ii}\vert=\vert \gamma_i a_{ii}-\rho (G_i+\sum_k a_{ki})\vert=\rho (G_i+\sum_k a_{ki})-\gamma_i a_{ii},
\end{equation}
because either $a_{ii}=0$, in which case the identity is obvious, or $a_{ii}>0$, in which case $\rho\sum_{k} a_{ki}> \sum_{k} a_{ki}\geq  a_{ii}\geq \gamma_i a_{ii}$.

Using \eqref{eq: id1},
we get
\[
\vert d_{ii}\vert=(\rho-\gamma_i)a_{ii}+\rho G_i+\rho\sum_{k\neq i}a_{ki}\geq (1-\gamma_i)a_{ii}+G_i>0,
\]
where the last inequality is due to either \eqref{yy1} or \eqref{yy2}. Thus $B(\rho)$ is strictly diagonally dominant, whence invertible, an impossibility. Thus we have proved $\rho(B)<1$.

The proof of \eqref{Skew6: pos} is similar to the proof of \eqref{Skew4: pos}, noting that due to irreducibility of $A$, strict diagonal dominance for one row (and weak diagonal dominance elsewhere) suffices to obtain the same conclusion (cf.~\cite[Theorem 6.2.27]{horn2012matrix}).

Proof of \eqref{Skew3: pos}: Finally, using the same argument as in the proof of \eqref{Skew2: pos}, we see that $B$ is similar to a stochastic matrix, hence $\rho(B)\leq 1$. 
\end{proof}

\begin{figure}
\includegraphics[width=0.8\textwidth]{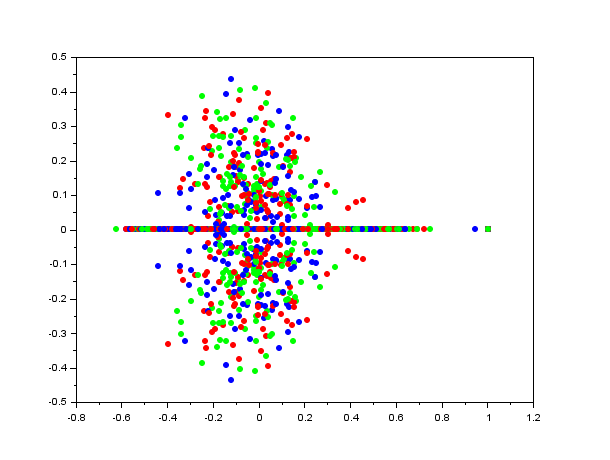}
\caption{The largest eigenvalue of skew-stochastic matrices equals one. (Theorem \ref{thm: skew stoch}). For the figure, we plotted the eigenvalues of skew-stochastic matrices $A$ in the complex plane, which are derived from $10000$ randomly generated $5\times 5$ matrices $A$ and $B$ are the column-sum scaled
matrices $A$, whose entries are independent, chi-square distributed with one degree of freedom.}
\label{figure1}
\end{figure}

\subsubsection{Characterization for Irreducible Matrices}
We have seen that, under mild non-degeneracy conditions, any non-negative skew-sub-stochastic matrix satisfies $\rho(B)\leq 1$. A. Neumaier (Vienna) conjectured
in private communication that the converse also holds.\footnote{I thank A. Neumaier for the proof of the implication \eqref{arnold1} $\Rightarrow$ \eqref{arnold2} in Theorem \ref{th: arnold}.} As we prove next,
this is indeed true under the additional assumption of irreducibility. However Example \ref{ex counter} shows that the assumption of irreducibility cannot, in general, be dropped.
\begin{theorem}\label{th: arnold}
Let $B$ be a non-negative irreducible matrix. The following are equivalent:
\begin{enumerate}
\item \label{arnold1} $\rho(B)\leq 1$.
\item \label{arnold2} $B$ is skew-sub-stochastic, that is: there exist $a_{ij}\geq 0$, $\gamma_i\in (0,1]$ and $G_i\geq 0$ ($1\leq i,j\leq n$) such that
\eqref{eq: Skew matrix} holds.
\end{enumerate}
Furthermore, $\rho(B)=1$ if and only if $\gamma_i=G_i=0$ for $1\leq i\leq n$.
\end{theorem}
\begin{proof}

The implicationn \eqref{arnold2} $\Rightarrow$ \eqref{arnold1} follows from Theorem \ref{thm: skew stoch} \ref{Skew2: pos},
as any irreducible matrix has non-vanishing column-sums.

We prove next \eqref{arnold1} $\Rightarrow$ \eqref{arnold2}: Let $s=(s^1,\dots,s^n)^\top$ be the the strict positive Perron eigenvector of $B^\top$ associated with $\rho:=\rho(B^\top)=\rho(B)$. Introduce $A=(a_{ij})_{1\leq i,j\leq n}$ by
$a_{ij}=b_{ij}s^i$ for $1\leq i,j\leq n$, and let $G_i=0$, $\gamma_i=\rho$ for $1\leq i\leq n$. Then indeed
\[
\frac{\gamma_i a_{ij}}{G_i+\sum_k a_{ki}}=\frac{\rho b_{ij}s^i}{\sum_k b_{ki}s_k}=\frac{\rho b_{ij} s^i}{\rho s^i}=b_{ij},
\]
and thus $B$ is skew-sub-stochastic.
Finally, if $\gamma_i=G_i=0$ for $1\leq i\leq n$, then $\rho(B)=1$, because $B$ is similar to a stochastic matrix. On the other hand, when $\rho(B)=1$, then by the proof above, $G_i=0$ and $\gamma_i=\rho(B)=1$ for $1\leq i\leq n$.
\end{proof}

The implication \eqref{arnold1} $\Rightarrow$ \eqref{arnold2} fails in general, if one removes irreducibility:

\begin{example}\label{ex counter}
For the matrix
\[
B=\left(\begin{array}{ll}1 & 0\\ 2 & 1\end{array}\right).
\]
the spectral radius $\rho(B)=1$. If $B$ were skew-(sub)stochastic, the following identities had to hold:
\begin{align}\label{eq1}
1&=\frac{\gamma_1 a_{11}}{G_1+a_{11}+a_{21}},\\\label{eq2}
 0&=\frac{\gamma_1 a_{12}}{G_1+a_{11}+a_{21}}\\\label{eq3}
 2&=\frac{\gamma_2 a_{21}}{G_2+a_{12}+a_{22}}\\\label{eq4}
 1&=\frac{\gamma_2 a_{22}}{G_2+a_{12}+a_{22}}
\end{align}

Starting with \eqref{eq4}, we see that due to $\gamma_2\leq 1$ and $G_2\geq 0$, we have $a_{12}=0$, $\gamma_2=1,G_2=0$. Therefore, \eqref{eq2} is automatically satisfied. Similarly, \eqref{eq1} implies $\gamma_1=1$ and $G_1=0$, as well as $a_{21}=0$, which violates \eqref{eq3}. We conclude that $B$ is {\it not} skew-sub-stochastic.
\end{example}
In the previous example, the spectral radius was $1$. Staying away from unit spectral radius allows the following conclusion:
\begin{theorem}
Let $B$ be a non-negative matrix such that $\rho(B)<1$. Then $B$ is skew-sub-stochastic, that is: there exist $a_{ij}\geq 0$, $\gamma_i\in (0,1]$ and $G_i\geq 0$ ($1\leq i,j\leq n$) such that
\eqref{eq: Skew matrix} holds.
\end{theorem}

If $B=0$, then we can choose $a_{ij}=0$ for $1\leq i,j\leq n$, and $\gamma_i=G_i=1$ for $1\leq i \leq n$. Next, assume
that $\rho(B)\in (0,1)$. For $\varepsilon>0$ we define $B^\varepsilon$ element-wise as $b_{ij}^\varepsilon:=b_{ij}+\varepsilon$ for $1\leq i,j\leq n$. Suppose $\varepsilon$ is sufficiently small such that $\rho^\varepsilon=\rho(B^\varepsilon)\leq 1$. Since $B^\varepsilon$ is strictly positive, it is irreducible. Let $s^\varepsilon$ be the strictly positive Perron vector of $(B^\varepsilon)^T$, that is $s_i^\varepsilon>0$ and $\sum_k b^\varepsilon_{ki} s_k^\varepsilon=\rho^\varepsilon s_i$ for $1\leq i\leq n$. Put
\[
\gamma_i:=\rho^\varepsilon,\quad G_i:=\varepsilon  \sum_k s_k^\varepsilon,\quad a_{ij}:=s_i^\varepsilon b_{ij}.
\]
Then
\[
\frac{\gamma_i a_{ij}}{G_i+\sum_k a_{ki}}=\frac{\rho^\varepsilon s_i^\varepsilon b_{ij}}{\varepsilon  \sum_k s_k^\varepsilon+\sum_k s_k^\varepsilon b_{ki}}=\frac{\rho^\varepsilon s_i^\varepsilon b_{ij}}{\sum_k s_k^\varepsilon b^\varepsilon_{ki}}=\frac{\rho^\varepsilon s_i^\varepsilon b_{ij}}{\rho^\varepsilon s_i^\varepsilon}=b_{ij},
\] 
and thus \eqref{eq: Skew matrix} holds

\subsubsection{Meme Popularity Subsides.}
For this section, we assume that
$\tau_a^i\sim \mathcal E(\beta_i)$, so that the first moments are, in view of \eqref{eq: super ratio} in combination with \eqref{eq: A1},
\begin{equation}\label{eq: A one network}
A_{ij}=\frac{(1-\mu_i)\lambda_{ij}\beta_i}{\mu_i\beta_i+\sum_{k=1}^n \lambda_{ki}\beta_k}.
\end{equation}
Note that, due to Theorem \ref{th: arnold}, any non-negative matrix with $\rho(A)\leq 1$ has a representation of the form \eqref{eq: A one network}, so that, from this perspective, Model 1 is maximally parameterized.
\begin{proposition}\label{prop: no final}
Any process $\mu(t)$ with branching mechanism \eqref{eq: super} contains no final classes.
\end{proposition}
\begin{proof}
With positive probability, particle $i$ has no descendants upon death. In fact,
the probability that any particle $i$ has no descendants is
\begin{align*}
\mathbb P[\tau^i>0,\nu^i=(0,\dots,0)]&=\int_0^\infty \kappa_0^i dG^i(u)+\int_0^\infty \kappa_i^i dG^i(u)\prod_{j}(1-\lambda_{ij})\\&\geq \int_0^\infty \kappa_0^i dG^i(u)>0.
\end{align*}
Since a final class has the property that the total number of particles of this class stays constant, $\mu(t)$ does not have any final particle class.
\end{proof}

\begin{theorem}\label{prop: subskew social}
The probability of extinction equals one. 
\end{theorem}
\begin{proof}
The spectral radius of the first moment matrix $A$ satisfies $\rho(A)\leq 1$, as $A$ in \eqref{eq: A one network} is a skew-sub-stochastic matrix (Theorem \ref{thm: skew stoch}). Furthermore, $\mu(t)$ has no final classes due to Proposition \ref{prop: no final}. Thus, by Theorem \ref{th: irr char disc timec}, the probability of extinction equals one.
\end{proof}
\subsubsection{Criticality}
The typical definition of critical behaviour includes the assumption of irreducibility of the branching process (cf.~ Remark \ref{rem: irr}). We therefore start by characterizing irreducibility and then, assuming irreducibility, we characterize critical, sub- and supercritical behaviour.
\begin{proposition}\label{prop: char irr social}
The following are equivalent:
\begin{enumerate}
\item $\mu(t)$ is irreducible,
\item The matrix $A$ is irreducible,
\item $\Lambda:=(\lambda_{ij})_{ij}$ is irreducible.
\end{enumerate}
\end{proposition}
\begin{proof}
The matrix of first moments is given element-wise by $A_{ij}=\int_0^\infty \kappa^i(u)dG^i(u)\times\lambda_{ij}$. Both $A$ and $\Lambda$ are non-negative matrices. Since the pre-factor $\int_0^\infty \kappa^i(u)dG^i(u)$ is strictly positive, and only scales rows, its elements are strictly positive if and only if $\Lambda$'s are. Therefore $A$ is irreducible if and only if $\Lambda$ is.
The rest of the claim follows from Theorem \ref{th: irr char disc time}.
\end{proof}

Second, we study the non-degenerateness of the second moment:
\begin{lemma}\label{lem: B2}
If $\mu(t)$ is irreducible, then for all $1\leq i\leq n$ we have $\sum_{k\neq i}\lambda_{ki}\neq 0$, and the second moment $B$ is non-zero.
\end{lemma}
\begin{proof}
Since $A$ is irreducible, also $\Lambda$ is (Proposition \ref{prop: char irr social}). Since $\lambda_{ii}=1$ for $1\leq i\leq n$, and since $\Lambda$ is irreducible, at least one row $i_0\in \{1,2,\dots,n\}$ of $\Lambda$ must have a
non-zero entry, besides the diagonal one, that is, there must exists $j_0\in\{1,2,\dots, n\}$
such that $\lambda_{ij}>0$. Let $u$ be the left eigenvector associated to $\rho(A)=1$,
and $v$ the right eigenvector of $A$ associated to $\rho(A)$. Since $A$ is a non-negative irreducible matrix, we can pick these eigenvectors to be strictly positive in each entry (\cite[Satz IV.5.4]{Sewastjanov}). Therefore, the second moments \eqref{eq: B2} satisfies
\begin{align*}
\sum_i\sum_{k,l} v^i b^i_{kl}u^k u^l&=\int_0^\infty \kappa^i(u)dG^i(u) \sum_{i,k,l} v^i \lambda_{ik}\lambda_{il}u^k u^l\\&\geq \int_0^\infty \kappa^i(u)dG^i(u) \lambda_{i_0 i_0}\lambda_{i_0 j_0} v^{i_0}u^{i_0}u^{j_0}>0.
\end{align*}
\end{proof}

By the proof of Proposition \ref{prop: subskew social} we know that $\rho(A)\leq 1$. We improve this statement in the following:
\begin{lemma}\label{prop: eigen 1 or not}
Suppose $\mu(t)$ is irreducible. Then,
\begin{enumerate}
\item \label{crit} $\rho(A)=1$, if $\mu_i=0$ for all $1\leq i\leq n$.
\item \label{subcrit} $\rho(A)<1$, if there exists $i\in\{1,\dots,n\}$ such that $\mu_i\in (0,1)$.
\end{enumerate}
\end{lemma}
Note that this is a full characterisation of the range of $\rho(A)$ for irreducible $A$, as innovation rates never equals one.
\begin{proof}
By Lemma \ref{lem: B2}, for all $1\leq i\leq n$ we have $\sum_{k\neq i}\lambda_{ki}\neq 0$. Therefore, if $\mu_i=0$ for all $1\leq i\leq n$, then $\rho(A)=1$ by Theorem \ref{thm: skew stoch} \eqref{Skew3: pos}.

Furthermore, by Proposition \ref{prop: char irr social}, the matrix $A$ is irreducible. Therefore,  if there exists $i\in\{1,\dots,n\}$ such that $\mu_i\in (0,1)$, then $\rho(A)<1$ by Theorem \ref{thm: skew stoch} \eqref{Skew6: pos}.
\end{proof}

A combination of Proposition \ref{prop: char irr social}, Lemma \ref{lem: B2} and Lemma \ref{prop: eigen 1 or not} gives the following:
\begin{theorem}\label{th: crit irr etc} 
Suppose $\Lambda=(\lambda_{ij})_{1\leq i,j\leq n}$ is irreducible. Then $\mu(t)$ is subcritical if and only if $\mu^i>0$ for some $1\leq i\leq n$, and critical, if and only if $\mu^i=0$ for all $1\leq i\leq n$. It is never supercritical.
\end{theorem}

We close this section by reflecting about the parameterisation of the model.

\begin{remark}\label{rem: parameterization}
Model 1 is maximally parameterized in that for any irreducible, non-negative matrix $A$ with spectral radius $\rho(A)=1$, there exists a (not unique) social network
having first moment matrix $A$, and where the entries of $A$ are proportional to the acceptance rates $\lambda_{ij}$, and the activites $\beta_i$ of the users are proportional to entry $s^i$ of the left Perron eigenvector of $A$.
More precisely, let $0\neq s\geq 0$ such that $A^\top s=s$, then it is well-known that $s^i>0$ for all $1\leq i\leq n$. As in the proof of Theorem \ref{th: arnold}, we can write
\[
a_{ij}=\frac{a_{ij}s^i}{\sum_{k}a_{ki}s^k}.
\]
Let $a_{\infty}=\max_{1\leq i,j\leq n} A_{ij}$. Then $\lambda_{ij}:=a_{ij}/a_{\infty}$ may be interpreted as acceptance probabilities, and, together with the intensities $\beta_i:=a_{\infty}\cdot s^i$, we have
\[
a_{ij}=\frac{\lambda_{ij}\beta_i}{\sum_{k}\lambda_{ki}\beta_k},
\]
which is the first moment matrix of a critical multiple-type branching process as Model 1, with zero innovation (cf. eq. \eqref{eq: A one network} and Theorem \ref{th: crit irr etc}. 
Note that if $\lambda_{ii}<1$, the model that restricts the amount of sharing of users $i$ accordingly is a simple generalization of Model 1, where $\lambda_{ii}=1$ for $1\leq i\leq n$.)
\end{remark}
\subsubsection{Model 1b: Multi-layer Version}\label{sec: multi}
The theory developed in this section can easily be extended to the multi-layer network of \cite{d2019spreading}, where sharing of information
between several distinct social media platforms is allowed. In the branching approximation of this multi-plex network model, the matrix of first moments \cite[Chapter 3, Equation (14)]{d2019spreading}
\begin{equation}\label{eq: multapp}
\frac{(c_{ij}+\lambda_{ij})(1-\mu_i) \beta_i}{\mu_i\beta_i+\sum_k \lambda_{ki}\beta_k+\sum_k c_{ki}\beta_k+\sum_k [\lambda_{ki}(\sum_l c_{lk}\beta_l)]}
\end{equation}
is of skew-(sub)stochastic form. Thus extinction and criticality can be characterized in the same fashion as above (Theorem \ref{prop: subskew social} and  Theorem \ref{th: crit irr etc}). There is, however, one slight difference: Due to the fact that a user can have multiple accounts $i,j$, and thus share with probability $c_{ij}$ some information they see on a different layer (one social media platform) with their follower on another layer (another social media app), the system is, in general, sub-critical, and not critical, even when all innovation rates vanish (this issue comes from the term $\sum_k [\lambda_{ki}(\sum_l c_{lk}\beta_l)]$ in \eqref{eq: multapp}). This is in contrast to the single layer Model 1 (The second part of Theorem \ref{th: crit irr etc} states that criticality holds precisely when all innovation rates are zero.)

It should be noted that the sub-criticality of a one-layer network has only been proved in \cite{d2019spreading} for sufficiently small innovation probability, and that our proofs concerning the spectrum of the first moment matrix
does not require irreducibility. In their multi-layer version, \cite{d2019spreading} require for the proof of sub-criticality the assumption not only of irreducibility (to identify uniquely left and right eigenvectors), asymptotically small innovation rates, and the assumption of a single, dominant layer. We have not used any of these assumptions to prove, along the lines of the previous section, $\rho(A)\leq 1$, and extinction in finite time.

\subsection{Model 2: Random Network}\label{Sec: Model2}

In this section we give a new description of information spreading through the random network model \cite{gleeson2016effects}, allowing unrestricted sharing
as in \cite{gleeson2014competition}.\footnote{For the sake of simplicity, we do not model memory effects here, which can be added without
any difficulty.} For the spreading of a special meme
through a social media network, we identify not users, but classes of users, whose instances are nodes in a directed random network, comprised of a possibly large, but finite, number $n\geq 2$ of types of nodes, each node $i$ with a fixed in-degree $z\geq 1$.\footnote{The original setup allows also an in-degree distribution, but due to the branching process approximation,
only its mean is of relevance. Interestingly, we prove in Section \ref{sec: tuning} below that the model in \cite{gleeson2016effects} requires a modulation of
meme arrival rates to have the criticality properties claimed in \cite{gleeson2016effects}.} The out-degree distribution $p_k>0$ ($1\leq k\leq n$) in this random network satisfies\footnote{Several parameter assumptions, in particular strict positivity of $p_k$, could be relaxed, at the expense of
sacrificing irreducibility of the branching process defined below.}, by assumption, $\sum_{k=1}^n k p_k=z$. We introduce the PGF 
\begin{equation}\label{eq: gp}
g(z):=\sum_{k=1}^n p_k z^k,\quad 0\leq z \leq 1.
\end{equation} 

Modelling the information spread in this network, we will insist on a branching mechanism, very similar to Model 1, and therefore we keep the description below brief. We consider a specific meme, which is identified with a Type $k$ particle, once it arrives in the stream of a Type $k$ user (since this is a random network, there
are possibly more than one Type $k$ users, but all share the same parameters set out below). Users of Type $k$ ($1\leq k\leq n$) get active at an exponentially distributed random time $\tau_k$ with parameter $\beta_k>0$. Furthermore, the arrival of information at the account of a user of Type $k$ is exogenous, in that memes are assumed to arrive according to a Poisson process with rate
\[
r_k=\mu_k \beta_k+ z\overline \beta-(1-\delta)\beta_k,
\]
where $\mu_k\in [0,1)$ is the innovation probability of user $k$, $\delta>0$ is a real parameter\footnote{This parameter will be decided later so to make the model critical in the case of zero innovation}, and
\[
\overline \beta:=\frac{\sum_{k} kp_{k} \beta_{k}\lambda_{k}}{\sum_k k  p_k}=\frac{1}{z}\sum_{k} kp_{k} \beta_{k}\lambda_{k},
\]
where $\lambda_i\in (0,1]$ denotes the probability of user $i$ accepting a meme into her stream. ($1\leq i\leq n$).

Let $\tau_k^o$ be the occupancy time of the meme in user $k$-th' account (the time between the arrival of other memes). Once a meme is accepted into the stream of a Type $k$ user, she becomes active at $\tau_k$, and the following mutually exclusive cases
occur: 
\begin{itemize}
\item If $\tau_k<\tau_k^o$, then either
\begin{itemize}
\item User $k$ innovates. This happens with probability $\mu_k$. Thus the particle dies without descendants.
\item User $k$ shares (with probability $1-\mu_i$). This means, that the meme is shared with $k$ random users. Besides that, the meme may either be deleted on her account ($\varepsilon=0$), or be shared once $(\varepsilon=1$). 
\end{itemize}
\item If $\tau_k>\tau_k^o$, the particle dies without descendants.
\end{itemize}

Let $F^i(s,t)$ be the PGF of a Type $i$ user (with $i$ followers), and let $h^k(s;u)$, where $s=(s^{i},1\leq i\leq n)$, be the PGF of the particle distribution of a Type $i$ user conditional on the death having occured by time $\tau^{i}=u$. We have
\[
h^{i}(s,u)=\kappa_0^{i}(u)+\kappa^{i}(u) (s^{i})^{\varepsilon}\left(\sum_{j=0}^k p_{i}\left( (1-\lambda_{j})+\lambda_{i} s^{i}\right)\right)^i,\quad 1\leq i\leq n,
\]
where $\varepsilon\in\{0,1\}$, and 
\[
\kappa^{i}(u):=(1-\mu_{i})e^{-r_{ij} u}, \quad\text{and}\quad \kappa^{i}_0(u)=\mu_{i}+(1-\mu_{i})(1-e^{-r_{i} u}).
\]
The unconditional particle distribution at death of particle $i$ is given by
\begin{equation}\label{eq: uncod Model2}
h^{i}(s)=\frac{r_i+\mu_i\beta_i}{\beta_i+r_i}+\frac{(1-\mu_{i})\beta_{i}}{\beta_{i}+r_{i}}(s^{i})^\varepsilon\left(\sum_{k} p_{k}\left( (1-\lambda_{k})+\lambda_{k}s^{k}\right)\right)^i,\quad 1\leq i\leq n,
\end{equation}
so that the matrix of first moments $A$ has the entries 
\begin{equation}\label{eq: 1}
A^{i}_{j}=\frac{(1-\mu_{i}) \beta_{i}}{\mu_i \beta_i+ z\overline \beta+\delta \beta_i}\times ( i p_{j}\lambda_{i}+\varepsilon \delta_{ij}),\quad 1\leq i, j\leq n.
\end{equation}
Moreover, the matrix of second moments is given by
\[
B^i_{kl}=\frac{(1-\mu_{i}) \beta_{i}}{\mu_i \beta_i+ z\overline \beta+\delta \beta_i}\left(i(i-1)(p_k\lambda_k)^2 +\varepsilon \delta_{ik}p_k\lambda_k+\varepsilon \delta_{il}p_l\lambda_l+\varepsilon(\varepsilon-1)\delta_{ik}\delta_{il}\right).
\]
\begin{lemma}\label{lem: non-deg}
$A$ (and thus the branching process) is irreducible, and the second moment is non-vanishing. Moreoever, if $\rho(A)=1$, then $\mu(t)$ is critical.
\end{lemma}
\begin{proof}
Due to the assumption $p_k>0$ for $1\leq k\leq n$, the matrix $A$ is strictly positive, and therefore irreducible. By Theorem \ref{th: irr char disc time}, the process $\mu(t)$ is also irreducible.

Since $n\geq 2$ and $\mu_i<1$, $\beta_i,\lambda_i,p_i>0$, we have 
\[
B^2_{kl}\geq \frac{(1-\mu_{i}) \beta_{i}}{\mu_i \beta_i+ z\overline \beta+\delta \beta_i}\left( 2(p_k\lambda_k)^2\right)>0,
\]
whence $B\neq 0$. Furthermore, if $A$ is irreducible, let $u$ (respectively $v$) being the strictly positive eigenvectors associated with $\rho(A)$, then $\sum_{i,k,l}v_i B^i_{kl}u^k u^l>0$. Thus, by Definition \ref{def: crit}, $\mu(t)$ is critical.
\end{proof}

\subsubsection{Extinction Probability with and without restricted sharing}
In this section we prove that in Model 1, extinction occurs with probability one, and characterize criticality. Note that for $\varepsilon=0$, the proofs are simpler and more instructive, and thus
we separate this case (Proposition \ref{prop?} from $\varepsilon\geq 1$ (Theorem \ref{th: 1x}). The following is elementary:
\begin{lemma}\label{eq: baby}
Let $v,w$ be non-zero $n$-vectors. Then the spectrum of the matrix $A=v w^\top$ is $\sigma(A)=\{0,v^\top w\}$, with $v$ (respectively, $w$) being a right (respectively, left) eigenvector associated with $\rho(A)=vw^\top$, and $w$.
\end{lemma}

\begin{proposition}\label{prop?}
Suppose $\varepsilon=0$, and for all $1\leq k\leq n$
\[
r_k=\mu_k \beta_k+ z\overline \beta-\beta_k,
\]
(that is, $\delta=\varepsilon=0$), then
\begin{enumerate}
\item  \label{part 1: eps0} If $\mu_i=0$ for all $1\leq i\leq n$, then $\rho(A)=1$  (and thus the process is critical.)
\item  \label{part 2: eps1} If $\mu_i>0$ for at least one $i\in\{1,\dots,n\}$, then $\rho(A)<1$  (and thus the process is subcritical.)
\end{enumerate}
\end{proposition}
\begin{proof}
The process is irreducible (Lemma \ref{lem: non-deg}). Obviously, matrix $A$ is of rank one. Therefore, we may use Lemma \ref{eq: baby} to compute the spectral radius of $A$. First, setting all innovation probability equals zero, we thus get
\[
\rho(A)=\frac{\sum_{i}i\beta_ip_i\lambda_i}{z\overline \beta}=1,
\]
which proves Part \ref{part 1: eps0}. Similarly, if one $\mu_i>0$, we get 
\[
\rho(A)\leq 1-\frac{\mu_i \beta_i i p_i \lambda_i}{\mu_i \beta_i+z \overline \beta}<0,
\]
which, in conjunction with Lemma \ref{lem: non-deg}, proves Part \ref{part 2: eps1}.
\end{proof}
We thus have an independent proof of the claim of criticality of the model \cite{gleeson2016effects}, when the in-degree distribution is degenerate. (For the more general situation,
one needs to modify the arrival rate of memes, see Section \ref{sec: tuning} below.)

Next we study, when $\varepsilon=1$. As the model is well-defined even for any $\varepsilon\in\mathbb N_0$, we state a more general version.
\begin{theorem}\label{th: 1x}
If $\delta=\varepsilon\in\{1,2,\dots\}$, that is,
\[f
r_k=\mu_k \beta_k+ z\overline \beta-(1-\varepsilon)\beta_k,\quad 1\leq k\leq n,
\]
the following hold:
\begin{enumerate}
\item  \label{part 11} If $\mu_i=0$ for all $1\leq i\leq n$, then $\rho(A)=1$ (and thus the process is critical.)
\item  \label{part 21} If $\mu_i>0$ for at least one $i\in\{1,\dots,n\}$, then $\rho(A)<1$ (and thus  the process is subcritical.)
\end{enumerate}
\end{theorem}
\begin{proof}
 First we show $\rho(A)\leq 1$ for all parameter choices. Assume, for a contradiction, that $\rho(A)>1$. We can write the matrix more generally as having entries
\[
\frac{\gamma_i\beta_i}{G_i+z\overline \beta+\delta \beta_i}(\varepsilon \delta_{ij}+ip_j\lambda_i),\quad 1\leq i,j\leq n,
\]
where $0<\gamma_i\leq 1$, and $G_i\geq 0$. By multiplying the matrix $A$ from the left by $\text{diag}(u)$, where  $u_i=(G_i+z\overline\beta+\delta\beta_i)^{-1}$, $1\leq i\leq n)$,
we see that $A-\rho I$ is singular if and only the matrix $R=(r_{ij})_{ji}$ is singular, which is defined element by element as
\[
\gamma_i\beta_i(\varepsilon \delta_{ij}+ip_j\lambda_i)-(G_i+z\overline \beta+\delta \beta_i)\rho\delta_{ij}.
\]
However, for any $1\leq i\leq n$, its diagonal element satisfies for $\delta=\varepsilon$
\[
\vert dr_{ii}\vert=\rho(G_i+z\overline \beta+\delta \beta_i)-\gamma_i\beta_i(\varepsilon +ip_i\lambda_i)>z\overline \beta-i \beta_ip_i\lambda_i=\sum_{j\neq i}j\beta_jp_j\lambda_j=\sum_{j\neq i} \vert r_{ij}\vert,
\]
whence the matrix $R$ is diagonally dominant and therefore non-singular, a contradiction.

Proof of \eqref{part 11}: We show that $1$ is an eigenvalue, when $\mu_i=0$ for $1\leq i\leq n$. By multiplying the matrix $A$ from the left by $\text{diag}(u)$, where  $u_i=(z\overline\beta+\delta\beta_i)^{-1}$, $1\leq i\leq n)$, we see that $A$ has eigenvalue one if and only if the matrix $R=(r_{ij})_{ji}$ is singular, where
\[
r_{ij}=i\beta_i \lambda_ip_j+\varepsilon\beta_i \delta_{ij}-\delta_{ij}(z\overline\beta+\delta\beta_i).
\]
If $\delta=\varepsilon$, then $\sum_j r_{ij}=0$. Hence the sum of all rows vanishes, which indeed makes $R$ singular. Since, in addition $A$ is non-negative  $\rho(A)\leq 1$, it follows that $\rho(A)=1$. The process
$\mu(t)$ is critical by Lemma \ref{lem: non-deg}.

Part \eqref{part 21} can be proved similarly, assuming, for a contradiction, $\rho(A)>1$. Note that for an irreducible matrix it suffices
to show weak diagonal dominance for all rows or columns, and strict diagonal dominance for a single row or column.
\end{proof}
The construction of meme arrival rates $r_k$ at accounts of users of Type $k$ (Proposition \ref{prop?} and \ref{th: 1x}) is intuitive: It needs to be increased by $(r-1)*\beta_k$, whenever the particle of type $k$ produces for sure
$r$ extra particles of the same type in the next generation.

The spectrum of first moment matrices $A$ is depicted in Figure \ref{figure2}, where we use numerical simulation to create a large amount of model parameters. In general, the spectrum  of $A$ is of size $n$: There are $n$ distinct eigenvalues, and all are real. This is in stark contrast to the situation of Model 1, where the first moments are skew-stochastic matrices, which
are are, in the irreducible case, similar to stochastic matrices (and thus have a complex spectrum, in general).

\begin{figure}
\includegraphics[width=0.8\textwidth]{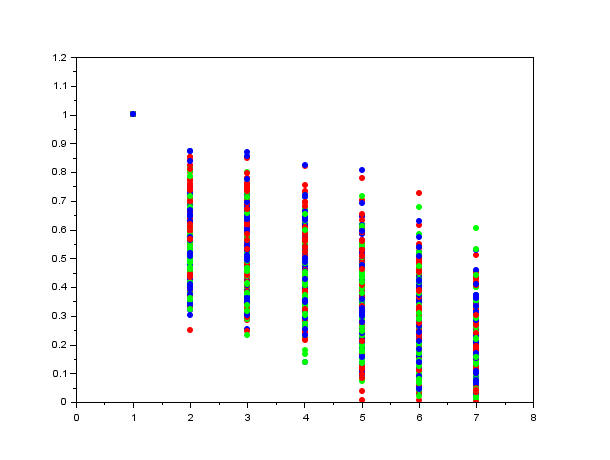}
\caption{The spectrum of 100 first moment matrices, using uniformly distributed parameters $\beta_k,\lambda_k, p_k$, $1\leq k\leq n$, while $\varepsilon=\delta$ is sampled from the Poisson distribution. For zero of innovation, the spectrum of each of these
is typically of size $n$, entirely real, with spectral radius equal one (Theorem \ref{th: 1x}).}
\label{figure2}
\end{figure}

\section{Moment Closure: Criticality and Popularity}\label{sec: epsilon 0 representation}
Probability generating functions are the main tool in the field of branching processes with discrete state space. The analysis gets very difficult in the case of multi-type processes, as it involves analytic functions in several variables. The main objective in the papers \cite{gleeson2014competition, gleeson2016effects, d2019spreading}
is to model the critical behaviour, and the popularity of memes (the total progeny of postings of a specified meme) in social media platforms. Unfortunately, there are few results available about the total progeny in multi-type processes. For example, for critical and subcritical systems, it is known \cite[Section 3]{good1960generalizations} that the total number $\nu^i$ of particles at the end, where $\mu(0)=\mu^i(0)=(\delta_{ij})_{1\leq j\leq n}$
where present at time $t=0$, has PGF
\[
\sum_\alpha q^i_\alpha s^\alpha
\]
where the coefficient $q_\alpha=\mathbb P[\nu^i=\alpha=(\alpha_1,\dots,\alpha_n)]$ is the coefficient of the term
\[
(s^1)^{\alpha_1}\dots (s^i)^{\alpha_i-1}\dots (s^n)^{\alpha_n}
\]
in the series expansion of
\[
(h^1)^{\alpha_1}\dots (h^n)^{\alpha_n}\det(Q),
\]
where $Q=(q_{ij})_{ij}$ is given by
\[
q_{ij}:=\delta_{ij}-\frac{s^i}{h^i(s)}\frac{\partial h^i}{\partial s^j}.
\]
But, for non-trivial models, the computation of these moments is not feasible. In this section we shall develop exact single-type representations of multi-type branching processes, so to obtain rigorous characterizations of criticality (Section \ref{sec: 41} and \ref{sec: tuning}) as well as reliable information concerning the tails  of
the popularity distribution (Section \ref{sec: 43}).
\subsection{Alternative Proof of Proposition \ref{prop?}}\label{sec: 41}
In the special case $\varepsilon=\delta=0$ of Model 2 we recover essentially the model of \cite{gleeson2016effects}, albeit with deterministic in-degree distribution (as opposed to general in-degree distribution with average in-degree $z$ in  \cite{gleeson2016effects}), however with the more flexible type-dependent acceptance rates $\lambda_k\in (0,1]$ (as opposed to constant $\lambda$ in  \cite{gleeson2016effects}). Note that we have setup our model using age-dependent branching processes, whereas the derivation of the governing
equations (e.g., popularity, or population) of \cite{gleeson2016effects} is more intuitive, and therefore reads different. 

We have seen that in the case $\varepsilon=\delta=0$ the criticality of the system is easy to derive, as the spectrum of the first moment matrix $A$ contains just $0$ and $\text{trace}(A)$. Now that the rank of $A$ is one, it is intuitive that the process is, essentially one-dimensional. We demonstrate this
for the choice $\lambda_k=1$, $1\leq k\leq n$:

Furthermore, we know by the proof of Theorem \ref{th: irr char disc timec} that, concerning extinction, the process is fully equivalent to the discrete version with unconditional particle distribution $h^i(s)$, $1\leq i\leq n$. Let $\widetilde F^i(t)$ denote the PGFs of the $i$th particle in the discrete equivalent.
Let
\[
\theta_i=\frac{(1-\mu_i)\beta_i}{\beta_i+r_i}, \quad \vartheta_i=1-\theta_i.
\]
and define the average $\widetilde F(t,s)=\sum_k p_k  F^k(t,s)$.  Then we have
\[
\widetilde F^i(t+1,s)=\vartheta_i+\theta_i \left(\sum_k p_k (1-\lambda_k+\lambda_k \widetilde F^k(t,s))\right)^i, \quad t=0,1,2,\dots,\quad 1\leq i\leq n,
\]
Averaging once more, we get
\[
\widetilde F(t+1,s)=h_0+f(\widetilde F(t,s)),\quad \widetilde F(t=0)=\widetilde s:=\sum_k p_k(1-\lambda_k+\lambda_k s^k),
\]
where
\[
h_0:=\sum_k p_k (1-\lambda_k+\lambda_k \vartheta_k)
\]
and the PGF $f(z)$ is defined by
\[
 f(z)=\sum_k p_k \theta_k ((1-\lambda_k)+\lambda_k z)^k.
\]
The PGF $\widetilde F$ can be interpreted as the PGF of a one-dimensional discrete branching process $m(t)$, $t=0,1,2,\dots$, with branching mechanism
\[
h_0+f(z).
\]
As 
\[
f'(0)=\lambda_1p_1\theta_1=\frac{\lambda_1p_1\beta_1 (1-\mu_1)}{\beta_1+r_1}=\frac{\beta_1 p_1(1-\mu_1)}{\mu_1\beta_1+z\overline \beta}\leq \frac{\lambda_1\beta_1p_1}{\sum_k, k\beta_k \lambda_k p_k}<1
\]
extinction occurs with probability $1$ if and only if
\[
f'(1)=\sum_k k p_k  \lambda_k\theta_k\leq 1.
\]
This condition is indeed fulfilled, as
\[
\sum_k k p_k \lambda_k \theta_k =\sum_k k p_k \lambda_k \frac{(1-\mu_k)\beta_k}{\mu_k\beta_k+z\overline \beta}\leq\sum_k k p_k \lambda_k\frac{\beta_k}{z\overline \beta}=1,
\]
by the definition of $\overline \beta$. Furthermore, the inequality is an equality, if and only if all users have zero innovation.(That is, $\mu_i=0$ for all $1\leq i\leq n$.) Therefore, we have
obtained an independent proof of Proposition \ref{prop?}. 

\subsection{Model 2 with flexible in-degree distribution}\label{sec: tuning}
In this section, we find that in the more general model of Model 2 as given by  \cite{gleeson2016effects}, with non-trivial in-degree distribution, the claimed criticality (sub-criticality) does not hold, in general. We rectify this problem in the end of the section by modulating the exogenous arrival rates of memes at each user's account.

For a type $(j,k)$ node with an in-degree of $j$, the arrival rate are defined as
\begin{equation}\label{eqnull}
r_j=\mu\beta_{jk}+j\overline \beta \lambda,
\end{equation}
where $\overline \beta=\sum_{j,k} \frac{k}{z} p_{jk} \beta_{jk}$. Memes therefore produce off-spring -- averaged over their lifetime -- according to the PGFs
\[
h_{jk}(s)=1-l_{jk}+l_{jk} (\sum_{lm} p_{lm} (1-\lambda+\lambda s^{lm})^k,
\]
where
\[
l_{jk}=\frac{\beta_{jk} (1-\mu)}{r_j+\mu \beta_{jk}}.
\]
The matrix of first moments is of the form $A^{jk}_{lm}$ 
\[
A^{jk}_{lm}=\frac{\lambda k (1-\mu) \beta_{jk} p_{lm}}{r_j}.
\]
\cite{gleeson2016effects} claims that for $\mu=0$ the branching number
\begin{equation}\label{eq: wrong}
\sum_{jk}\frac{j}{z}\frac{\beta_{jk}\lambda k p_{jk}}{r_j}=1
\end{equation}
implies the criticality of the system.\footnote{Note that this claim amounts to saying that extinction in finite time can be proved using the one-dimensional stochastic process
\[
\mu(t):=\sum_{j,k}\frac{j}{z} p_{jk} \mu^{jk}(t)
\]
with first moment equals to the left side of \eqref{eq: wrong} (cf. \cite[Equation (7)]{gleeson2016effects})}

However, the spectral radius of $A$ is given by its trace: For innovation probability $\mu=0$ we get
\begin{equation}\label{eq: superdyn1}
\rho(A)=\sum_{jk}\frac{\lambda k  \beta_{jk}p_{jk}}{r_j}=\sum \frac{z}{j}\frac{\beta_{jk} k p_{jk}}{\sum \beta_{jk} k p_{jk}}
\end{equation}
which may even exceed one. (In this case the process would be supercritical, an undesirable feature.) In fact, only for
deterministic in-degree equals $z$, the critically is obvious, but this case has been studied in Model 2 already.

Nevertheless, by modulating the arrival intensities $r_{jk}$ appropriately, one can correct the model \cite[Equation (7)]{gleeson2016effects})
so to make it critical ($\mu=0$) or subcritical $(\mu>0$). Keeping the functional form \eqref{eqnull}, we can modify
the definition of $\overline\beta$, setting
\[
\overline\beta=\sum_{j,k}\frac{k}{j}\beta_{jk} p_{jk}.
\]
Note the slight difference with the original definition of $\overline \beta$, where division was by mean in-degree instead of actual degree. This change amounts to modulating the meme arrival rate as
\[
r_j=\mu \beta_{jk}+j\left(\sum_{j,k}\frac{k}{j}\beta_{jk} p_{jk}\right)\lambda.
\]
By the first identity in \eqref{eq: superdyn1}, $\rho(A)=1$ and thus, by Theorem \ref{th: irr char disc timec}, extinction occurs with probability one. (Furthermore, similarly to Model 2,  irreducibility of $A$ and non-vanishing second moment $B$
implies criticality.)
\subsection{Quality of Approximation}\label{sec: 43}
The approach of Sections \ref{sec: 41} and \ref{sec: tuning} to represent the multi-type branching model by a single-type branching process does not work when $\varepsilon\neq 0$ in Model 2, because in this case the dynamics of the averaged PGF $\tilde F:=\sum_k p_k F^k$ does not become a uni-variate recursion. Nevertheless \cite{gleeson2014competition} studies in a simple variant of \cite{gleeson2016effects}) such a uni-variate process as an approximation. (See also Section S1 in the supplementary material, \url{https://arxiv.org/pdf/1305.4328}.)

We study in this section the question raised by \cite{gleeson2014competition} concerning the quality of the approximation. To this end, we
are developing a new representation of the model, where the timing of sharing is speeded up, while the popularity distribution (that is, the total number of particles produced by the time of extinction)
is exact (Model C below).

\begin{itemize}
\item [Model A] \label{a} (Original model of \cite{gleeson2014competition}) Consider a discrete-time, homogenous version
of Model 2  (where $\beta_k=\beta$, $\mu_k=\mu$ and $\lambda_k=1$  for all $1\leq k\leq n$) , defined by the PGFs in \eqref{eq: uncod Model2} \footnote{This is, therefore, essentially the model studied in \cite{gleeson2014competition}, without the option of dealing answering messages. We remove this scenario, so to be consistent with Model 2 and thus with \cite{gleeson2016effects}, but the same analysis for obtaining the tail behaviour of
the popularity distribution can be used also for the most general version of \cite{gleeson2014competition}.}. With these simplifications, the PGFs in  \eqref{eq: uncod Model2} assume the simple form
\begin{equation}\label{eq: many}
h^i(s)=\eta+\zeta s^i  (g(s))^i,\quad 1\leq i\leq n,
\end{equation}
where the PGF $g$ is defined in \eqref{eq: gp},
\[
\eta:=\frac{2\mu+z}{\mu+z+1},\quad \zeta:=1-\eta=\frac{(1-\mu)}{\mu+z+1} 
\]
and 
\[
z=g'(1)=\sum_k k p_k.
\]
\item  [Model B] \label{B} (Single-type approximation of \cite{gleeson2014competition})  The single-type branching process $\tilde \mu(t)$, approximating\footnote{Formally, one can get this approximation by replacing
\[
s^k\left(\sum_i p_i s^i\right)^k\approx \left(\sum_i p_i s_i\right)^{k+1}
\]
in which case \eqref{eq: many} simplifies to
\begin{equation}\label{eq: one}
h^i(s)=\eta+\zeta  (\sum_k p_k s^k)^{i+1},\quad 1\leq i\leq n.
\end{equation}
Let $\mu^i(t)$ be the particle decomposition at time $t$ of a branching process that starts with precisely one particle of type $i$
and produces off-spring according to PGF \eqref{eq: one}. Thus, the single-type process $\tilde\mu(t)$ defined by
\begin{equation}\label{eq: superx}
\tilde \mu(t):=\sum_k p_k \mu^k(t)
\end{equation}
starting with a single particle, represented by the random draw from $(p_k)$, has precisely the branching mechanism defined by \eqref{eq: Macsi}. \cite{gleeson2014competition} first considers the dynamics of the total number of particles produced
by time $t\geq 0$ (the so-called tree-size) by the process on the right hand side of \eqref{eq: superx}, then approximates the resulting differential equation of the tree-sizes so to make it autonomous equation in the tree-size of $\tilde \mu$, not the individual
trees coming from $\mu^i$.} Model A, is defined by the PGF
\begin{equation}\label{eq: Macsi}
\widetilde h(x)=\eta+\zeta xg(x),\quad 0\leq x\leq 1.
\end{equation}

\item  [Model C] \label{C} (New, exact single-type representation) We define the single-type discrete-time process $\mu^*(t)$ by its one-period PGF
\[
h^*(x):=\eta \sum\frac{p_i}{1-\zeta x^i)},\quad 0\leq x\leq 1.
\] 
\end{itemize}
We have the following:
\begin{theorem}
The total progeny of $\nu$ of Model A, starting with the initial distribution $\nu_0$ drawn from $(p_k)_k$, equals in law
to the total progeny of Model C, starting with a single particle.
\end{theorem}
\begin{proof}
First, Model A's total population doesn't depend on timing of sharing in the following sense. Recall that in Model A, the following scenarios occur to a type $i$ particle in one step.
\begin{itemize}
\item Scenario 1 With probability $\eta$ it dies.
\item Scenario 2 With probability $\zeta=1-\eta$ the associated meme is shared: That is, the particle of Type $i$ is reproduced, so to live for one more period, and at the same time, $i$ random particles are produced, each of type $k$ with probability $k$.
\end{itemize}
The total progeny at the time of extinction is equivalent to a model $A1$, where we trace the future of the reproduced particle $i$ of Scenario 2, so to consider its descendants
being produced immediately (as opposed to being reproduced in the  subsequent period): A single meme at a user's account of Type $i$ thus produces
descendants in a single period as follows:
\begin{itemize}
\item Scenario 1 With probability $\eta$ it dies.
\item Scenario 2a With probability $\zeta\eta$ only $i$ descendants are produced, each of type $k$ with probability $k$.
\item Scenario 2b With probability $\zeta(1-\eta)$ it reproduces, thus another particle of Type $i$ lives for another period, and $2\cdot i$ random particles are produced, each of type $k$ with probability $k$.
\end{itemize}
This change of timing amounts to accelerating sharing, and can be expressed by partially iterating the PGF. By repeatedly doing so,  we obtain the following sequence of PGFs
\begin{align*}
&\eta+\zeta s^i  (\sum_l p_l s^l)^i,\\
&\eta+\zeta \left(\eta+\zeta s^i  (\sum_l p_l s^l)^i\right)  (\sum_l p_l s^l)^i,\\
&\quad\vdots\\
&\eta\sum_{k=0}^\infty \zeta^k  (\sum_l p_l s^l)^{ki}
\end{align*}
and thus, in law these processes converge to a branching processes with a one-period branching mechanism expressed by the PGF
\[
h^{i,*}(s):=\frac{\eta}{1-\zeta  (\sum_l p_l s^l)^i}.
\]
Defining the single variable $x:=\sum_l p_l s^l$, and the uni-variate PGF
\[
\tilde h^*(x):=\sum_i p_i h^{i,*}(s),
\]
we have a well-defined single-type branching process $\mu^*(t)$ , with branching mechanism is given by the PGF
\[
h^*(x)=\eta\sum_{i=1}^n \frac{p_i}{1-\zeta x^i}.
\]
\end{proof}

We are thus prepared to quantify the quality of the approximation of Model B, by comparison with Model C. To this end, we establish the asymptotic behaviour of the tails
of the popularity distribution in both models:\footnote{Note that the mathematical machinery we are using below is partly different to  \cite{gleeson2014competition}, as we do not model tree-sizes (population) directly, but
model the underlying branching process.}, using the following auxiliary statement (see \cite[Appendix S3, Lemma 1 and its proof]{gleeson2014competition}):
\begin{lemma}\label{lem: wilf}
Let $\Phi(x)=\sum_k \pi_k x^k$ be the PGF of the distribution $\pi_k$ and suppose $\Phi$ has the following
asymptotic series near $x=1$,
\begin{equation}\label{eq: tara}
\Phi(1-w)=\text{analytic part}\quad+\sum_{m=1}^\infty c_m w^{\beta_m},\quad w\rightarrow 0,
\end{equation}
where $w=1-x$ and $\beta_1<\beta_2<\dots$ are positive, non-integer powers. Then the leading order asymptotic behaviour of $\pi_k$ is
\[
\pi_k\sim \frac{c_1}{\Gamma(-\beta_1)}k^{-\beta_1-1},\quad k\rightarrow \infty,
\]
where $\Gamma$ denotes the Gamma function.
\end{lemma}

The PGF $\varphi(u)=\sum_{m=1}^\infty q_m u^m$ of the popularity distribution $q_m$, $m\geq $ of a single-type branching process in discrete time with branching mechanism $F(u)$, $0\leq u\leq 1$ satisfies
by \cite[Chapter V.5]{Sewastjanov}
\begin{equation}\label{eq: x}
\varphi(u)=u\cdot F(\varphi(u)).
\end{equation}
To get an expression of the form \eqref{eq: tara} for Model 2, we first approximate $F$ by a Taylor approximation of second order (using $F(1)=F'(1)=1$ and $F''(1)>0$, as satisfied by Model 2 and Model 3), and then substitute $\varphi(u)=1-\theta(w)$, where $w=1-u$, such that \eqref{eq: x} reads
\[
1-\theta\simeq (1-w)\left(1-\theta+\frac{B}{2}\theta^2\right).
\]
Hence, near $w=0$,
\[
\varphi(1-w)\simeq 1-\theta(w)\simeq -\sqrt{\frac{2}{F''(1)}}w^{1/2}+\text{analytic part}.
\]
Thus by Lemma \ref{lem: wilf}, noting that $\Gamma(-1/2)=-2\sqrt\pi$,
\begin{equation}\label{eq: atail}
q_m\sim \frac{m^{-3/2}}{\sqrt{2\pi F''(1)}}\quad \text{as}\quad m\rightarrow \infty.
\end{equation}
This power tail behaviour of order $-3/2$ agrees with the findings of  \cite{gleeson2014competition} in the same case (zero innovation, finite second moment of the degree distribution $p_k$), and therefore, one could be led to conclude that the approximation of Model A by Model C is indeed excellent. However, as shown above, this asymptotics is a feature of all critical single-type branching processes (that is, $F(1)=F'(1)=1$ with $F''(1)>0$). But not every uni-variate approximation of Model A with $\mu=0$ that is critical itself, is a good approximation in this sense. To give a meaningful measure of the quality of approximation, the precise rate in \eqref{eq: atail} can be compared, which depends on the constant $1/\sqrt{2\pi F''(1)}$, and thus involves the second moment $F''(1)$. We have in Model B
\[
F''(1)=\tilde h''(1)=\frac{g''(1)}{z+1}+\frac{z}{z+1},
\]
while the exact second moment in the exact representation C is
\[
F''(1)=(h^*)''(1)=\frac{1}{z}g''(1)+\frac{2}{z}.
\] 
\cite{gleeson2014competition} finds through simulation that when $p_k\sim k^{-\gamma}$ ($\gamma=2.5$, $k\geq 4$), the approximation Model B is satisfactory. Indeed, in this case
the mean degree $z\approx 10.4$, and the second moment of the out-degree distribution $g''(1)\approx 1941.3$, hence
\[
(h^*)''(1)-\tilde h''(1)\approx186.8-171.1=15.7 
\]
which amounts to an error of $8\%$ in the second moment, and thus the tails are of the popularity distribution are similar, but not as fat as those suggested by Model B.

As Model C provides an exact representation of the popularity of Model A, one can use standard results for one-dimensional branching processes to compute the exact form of the popularity distribution, at least numerically: By \cite[Satz V.5.4]{Sewastjanov},
\[
q_m=\frac{1}{m}\mathbb P[\xi_1+\dots+\xi_m=m-1],
\] 
where $\xi_i$ are independent, identically distributed random variables with PGF $F(u)$. In other words, $q_m$ is the $(m-1)$--ths coefficient
in the series representation of $F^m$, divided by $m$. Obviously, this is much easier to deal with than the aforementioned multi-variate version \cite{good1960generalizations}. Not surprisingly, Figure \ref{figure3} confirms that the asymptotic formula fore $q_m$ is better, the larger $m$ is, converging slowly to the true popularity distribution. Furthermore, it demonstrates that Model B is an excellent approximation of Model A for the specific parameter choice, as not only the asymptotics, but also the true popularity distributions agree well. 
\begin{figure}
\includegraphics{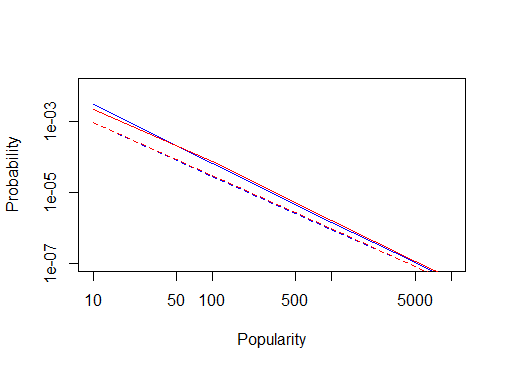}
\caption{The probaiblities $q_m$ in Model C (solid blue line) vs Model B (solid red line), where we use $p_k\sim k^{-\gamma}$, for $\gamma=2.5$, $k\geq 4$. The approximation \eqref{eq: atail} is in dashed lines. For these parameter choices, there is an excellent agreement between the uni-variate proxy of Model A by Model B, and the true uni-variate representation of Model A by Model C.}
\label{figure3}
\end{figure}

\bibliographystyle{plainnat}

\end{document}